\documentclass{ws-ijtaf}
\usepackage{graphicx,dsfont,amsmath,amssymb,subfigure}
\usepackage{color}

\newcommand{\e}{\mathrm{e}}

\newcommand{\roc}{\zeta^{\scriptscriptstyle (2)}}
\newcommand{\roa}{\zeta^{\scriptscriptstyle (\text{1})}}
\newcommand{\road}{\zeta^{\scriptscriptstyle (\text{d})}}
\newcommand{\rod}{\rho^{\scriptscriptstyle (\text{d})}}
\newcommand{\ros}{\rho^{\scriptscriptstyle (\text{s})}}

\newcommand{\tauK}{\tau^{\scriptscriptstyle (\text{K})}}
\newcommand{\sign}{\operatorname{sign}}

\newcommand \be  {\begin{equation}}
\newcommand \bea {\begin{eqnarray} \nonumber }
\newcommand \ee  {\end{equation}}
\newcommand \eea {\end{eqnarray}}
\newcommand{\tUU}{\tau_{\scriptscriptstyle UU}}
\newcommand{\tUL}{\tau_{\scriptscriptstyle UL}}
\newcommand{\tLU}{\tau_{\scriptscriptstyle LU}}
\newcommand{\tLL}{\tau_{\scriptscriptstyle LL}}

\begin{document}

\markboth{R.~Chicheportiche \and J.-P.~Bouchaud}
{The joint distribution of stock returns is not elliptical}

\catchline{}{}{}{}{}

\title{THE JOINT DISTRIBUTION OF STOCK RETURNS IS NOT ELLIPTICAL}

\author{R\'EMY~CHICHEPORTICHE}

\address{Capital Fund Management, 6--8 boulevard Haussmann\\
Paris, 75009, France\\ 
\email{Remy.Chicheportiche@cfm.fr} }

\author{JEAN-PHILIPPE~BOUCHAUD}

\address{Capital Fund Management, 6--8 boulevard Haussmann\\
Paris, 75009, France\\
Jean-Philippe.Bouchaud@cfm.fr}

\maketitle

\begin{history}
\received{6 September 2010}
\accepted{21 November 2011}
\end{history}

\begin{abstract}
Using a large set of daily US and Japanese stock returns, we test in detail the relevance of Student models,
and of more general elliptical models, for describing the joint distribution of returns. 
We find that while Student copulas provide a good approximation for strongly correlated pairs of stocks, 
systematic discrepancies appear as the linear correlation between stocks decreases, that rule out {\it all} elliptical models.  
Intuitively, the failure of elliptical models can be traced to the inadequacy of the assumption of a single volatility mode for all
stocks. 
We suggest several ideas of methodological interest to efficiently visualise and compare different copulas. 
We identify the rescaled difference with the Gaussian copula and the central value of the copula as strongly discriminating observables. 
We insist on the need to shun away from formal choices of copulas with no financial interpretation.
\end{abstract}

\keywords{Copulas; Stock returns; Multivariate distribution; Linear correlation; Non-linear dependences; Student distribution; Elliptical distributions}

\section{Introduction}

The most important input of portfolio risk analysis and portfolio optimisation 
is the correlation matrix of the different assets. 
In order to diversify away the risk, one must avoid allocating on bundles of correlated assets, 
an information in principle contained in the correlation matrix. 
The seemingly simple mean-variance Markowitz program is however well known to be full of thorns. 
In particular, the empirical determination of large correlation matrices turns out to be difficult, 
and some astute ``cleaning'' schemes must be devised before using it for constructing optimal allocations \cite{ledoit2004well}. 
This topic has been the focus of intense research in the recent years, some inspired by Random Matrix Theory 
(for a review see \cite{elkaroui2009,potters2009financial}) or clustering ideas \cite{marsili2004,tola2008cluster,tumminello2007hierarchically}. 

There are however many situations of practical interest where the (linear) correlation matrix is 
inadequate or insufficient \cite{bouchaud2003theory,malevergne2006extreme,mashal2002beyond}. 
For example, one could be more interested in minimising the probability of large negative returns 
rather than the variance of the portfolio. 
Another example is the Gamma-risk of option portfolios, where the correlation of the squared-returns of the underlyings is needed. 
Credit derivatives are bets on the probability of simultaneous default of companies or of individuals; 
again, an estimate of the correlation of tail events is required (but probably very hard to ascertain empirically) 
\cite{frey2001modelling}, and for a recent interesting review \cite{brigo}. 

Apart from the case of multivariate Gaussian variables, the description of non-linear dependence 
is not reducible to the linear correlation matrix. The general problem of parameterising the full 
joint probability distribution of $N$ random variables can be ``factorised'' into the specification 
of all the marginals on the one hand, and of the dependence structure (called the `copula') of 
$N$ standardised variables with uniform distribution in $[0,1]$, on the other hand. 
The nearly trivial statement that all multivariate distributions can be represented in that way is 
called Sklar's Theorem \cite{embrechts02correlation,sklar1959fonctions}. 
Following a typical pattern of mathematical finance, the introduction of copulas ten years ago 
has been followed by a calibration spree, with academics and financial engineers frantically looking 
for copulas to best represent their pet multivariate problem. 
But instead of trying to understand the economic or financial mechanisms that lead to some particular 
dependence between assets and construct copulas that encode these mechanisms, the methodology has been
--- as is sadly often the case --- to brute force calibrate on data copulas straight out from statistics 
handbooks \cite{durrleman2000copula,fermanian2005some,fischer2009empirical,fortin2002tail,patton2001estimation}.
The ``best'' copula is then decided from some quality-of-fit criterion, irrespective of whether the 
copula makes any intuitive sense (some examples are given below). 
This is reminiscent of the `local volatility models' for option markets \cite{dupire1994pricing}: 
although the model makes no intuitive sense and cannot describe the actual dynamics of the underlying asset, 
it is versatile enough to allow the calibration of almost any option smile (see \cite{hagan2002woodward}). 
This explains why this model is heavily used in the financial industry. Unfortunately, a blind calibration of 
some unwarranted model (even when the fit is perfect) is a recipe for disaster. 
If the underlying reality is not captured by the model, it will most likely derail in rough times --- a particularly 
bad feature for risk management! Another way to express our point is to use a Bayesian language: 
there are families of models for which the `prior' likelihood is clearly extremely small 
--- we discuss below the case of Archimedean copulas. 
Statistical tests are not enough --- intuition and plausibilty are mandatory. 

The aim of this paper is to study in depth the family of elliptical copulas, in particular Student copulas, 
that have been much used in a financial context and indeed have a simple intuitive interpretation 
\cite{embrechts02correlation,frahm2003elliptical,hult2002multivariate,luo2009t,malevergne2006extreme,shaw2007copula}. 
We investigate in detail whether or not such copulas can faithfully represent the
joint distribution of the returns of US stocks. (We have also studied other markets as well). 
We unveil clear, systematic discrepancies between our empirical data and the corresponding predictions of elliptical models. 
These deviations are qualitatively the same for different tranches of market capitalisations, 
different time periods and different markets. 
Based on the financial interpretation of elliptical copulas, we argue that such discrepancies are actually expected, 
and propose {the ingredients of} a generalisation of elliptical copulas 
that should capture adequately non-linear dependences in stock markets. 
The full study of this generalised model, together with a calibration procedure and stability tests, 
will be provided in a subsequent publication.

The outline of the paper is as follows. 
We first review different measures of dependence (non-linear correlations, tail dependence), including copulas. 
We discuss simple ways to visualise the information contained in a given copula. 
We then introduce the family of elliptical copulas, with special emphasis on Student copulas and log-normal copulas, 
insisting on their transparent financial interpretation, in contrast with the popular but rather implausible Archimedean copulas. 
We then carefully compare our comprehensive empirical data on stock markets with the predictions of elliptical models 
--- some of which being parameter free, 
and conclude that elliptical copulas fail to describe the full dependence structure of equity markets. 

\section{Bivariate measures of dependence and copulas}

In this section, we recall several bivariate measures of dependence, and study their properties 
when the distribution of the random pair under scrutiny is elliptical. The incentive to focus on bivariate measures comes from the theoretical 
property that all the marginals, including bivariate, of a multivariate elliptical distribution are themselves elliptical.
In turn, a motivated statement that the pairwise distributions are not elliptical is enough to claim the non-ellipticity of the joint multivariate distribution.
This is the basic line of argumentation of the present paper. 

\subsection{Correlation coefficients}

Beyond the standard correlation coefficient, one can characterise the dependence structure 
through the correlation between powers of the random variables of interest:
\begin{align}
	\rod_{ij}&=\frac{\mathds{E}[S_i S_j |X_i X_j|^d]-\mathds{E}[S_i|X_i|^d]\mathds{E}[S_j|X_j|^d]}{\sqrt{\mathds{V}[S_i|X_i|^d]\,\mathds{V}[S_j|X_j|^d]}}\\
	\road_{ij}&=\frac{\mathds{E}[|X_iX_j|^d]-\mathds{E}[|X_i|^d]\mathds{E}[|X_i|^d]}{\sqrt{\mathds{V}[|X_i|^d]\,\mathds{V}[|X_j|^d]}},
\end{align}
where $S_i = \sign(X_i)$ and provided the variances $\mathds{V}$ in the denominators are well defined. 
The case $\rod$ for $d=1$ corresponds to the usual linear correlation coefficient, 
for which we will use the standard notation $\rho$, 
whereas $\road$ for $d=2$ is the correlation of the squared returns, 
that would appear in the Gamma-risk of a $\Delta$-hedged option portfolio. 
However, high values of $d$ are expected to be very noisy in the presence of power-law tails 
(as is the case for financial returns) and one should seek for lower order moments, 
such as $\road$ for $d=1$ which also captures the correlation between the {\it amplitudes} (or the volatility) of price moves, 
or even $\ros \equiv \rod$ for $d=0$ that measures the correlation of the signs.

In the case of bivariate Gaussian variables with zero mean, all these correlations can be expressed in terms of $\rho$. 
For example, the quadratic correlation $\roc = \road$ for $d=2$ is given by:
\be
\roc = \rho^2;
\ee
whereas the correlation of absolute values $\roa=\road$ for $d=1$ is given by:
\be
\roa = \frac{D(\rho)-1}{\frac{\pi}{2}-1},
\ee
with $D(r)=\sqrt{1-r^2}+r\,\arcsin r$. The correlation of signs is given by \mbox{$\ros= \frac{2}{\pi} \arcsin \rho$}. 
For some other classes of distributions, the higher-order coefficients $\rod$ and $\road$ are explicit functions of the
coefficient of linear correlation. This is for example the case of Student variables (see Fig.~\ref{fig:dep}) 
and more generally for all elliptical distributions (see below).

\subsection{Tail dependence}

Another characterisation of dependence, of great importance for risk management purposes, 
is the so-called \emph{coefficient of tail dependence} which measures the joint probability of extreme events. 
More precisely, the upper tail dependence is defined as \cite{embrechts02correlation,malevergne2006extreme}:
\begin{equation}\label{eq:def_tauUU}
	\tUU(p)=\mathds{P}\left[X_i>\mathcal{P}_{\negmedspace{\scriptscriptstyle <},i}^{-1}(p) \mid{}X_j>\mathcal{P}_{\negmedspace{\scriptscriptstyle <},j}^{-1}(p) \right],
\end{equation}
where $\mathcal{P}_{\negmedspace{\scriptscriptstyle <},k}$ is the cumulative distribution function (cdf) of $X_k$, 
and $p$ a certain probability level. In spite of its seemingly assymetric definition, it is easy to show that 
$\tUU$ is in fact symmetric in $X_i\leftrightarrow{}X_j$ when all univariate marginals are identical.
When $p \to 1$, $\tUU^*=\tUU(1)$ measures the probability that $X_i$ takes a very 
large positive value knowing that $X_j$ is also very large, and defines the asymptotic tail dependence.
Random variables can be strongly dependent from the point of view of linear correlations, 
while being nearly independent in the extremes. 
For example, bivariate Gaussian variables are such that $\tUU^*=0$ for any value of $\rho < 1$. 
The lower tail dependence $\tLL$ is defined similarly:
\begin{equation}\label{eq:def_tauLL}
	\tLL(p)=\mathds{P}\left[X_i<\mathcal{P}_{\negmedspace{\scriptscriptstyle <},i}^{-1}(1-p) \mid{}X_j<\mathcal{P}_{\negmedspace{\scriptscriptstyle <},j}^{-1}(1-p) \right],
\end{equation}
and is equal to $\tUU(p)$ for symmetric bivariate distributions. 
One can also define mixed tail dependence, for example:
\begin{equation}\label{eq:def_tauUL}
	\tLU(p)=\mathds{P}\left[X_i<\mathcal{P}_{\negmedspace{\scriptscriptstyle <},i}^{-1}(1-p) \mid{}X_j>\mathcal{P}_{\negmedspace{\scriptscriptstyle <},j}^{-1}(p) \right],
\end{equation}
with obvious interpretations. 

\subsection{Copulas}\label{ssec:copulas}

There are many other possible coefficients of dependence, such as Spearman's rho or Kendall's tau, that measure, 
respectively, the correlation of ranks or the ``concordance'' probability (see e.g. \cite{numericalrecipes} for an introduction). 
Both these measures are invariant under any increasing transformations. 
More generally, the copula (generalisable to dimensions larger than 2), 
encodes all the dependence between random variables that is invariant under increasing transformations.
The copula of a random pair $(X_i,X_j)$ is the joint cdf of $U_i=\mathcal{P}_{\negmedspace{\scriptscriptstyle <},i}(X_i)$ 
and $U_j=\mathcal{P}_{\negmedspace{\scriptscriptstyle <},j}(X_j)$:
\begin{equation}\label{eq:def_copule}
	C(u_i,u_j)= \mathds{P}\left[\mathcal{P}_{\negmedspace{\scriptscriptstyle <},i}(X_i)\leq{}u_i 
	            \textrm{ and }  \mathcal{P}_{\negmedspace{\scriptscriptstyle <},j}(X_j)\leq{}u_j\right]
\end{equation}
Since the marginals of $U_i$ and $U_j$ are uniform by construction, the copula only captures their degree of ``entanglement''. 
For independent random variables, $C(u_i,u_j) \equiv u_iu_j$. 

Whereas the $\rod$'s and $\road$'s depend on the marginal distributions, the tail dependences, 
Spearman's rho and Kendall's tau can be fully expressed in terms of the copula only. 
For example, Kendall's tau is given by $\tauK=4\int_{[0,1]^2} C(u_i,u_j) {\rm d}C(u_i,u_j)-1$, 
while the tail dependence coefficients can be expressed as:
\begin{subequations}\label{eq:tau_C}
\begin{align}
	       \frac{1-2p+C(p,p)}{1-p}&=\tUU(p)  & 
	\frac{C(1\!-\!p,1\!-\!p)}{1-p}&=\tLL(p) \\ 
	  \frac{1-p-C(p,1\!-\!p)}{1-p}&=\tUL(p)  & 
	  \frac{1-p-C(1\!-\!p,p)}{1-p}&=\tLU(p).
\end{align}
\end{subequations}

Copulas are not easy to visualise, first because they need to be represented as 3-D plot of a surface in two dimensions 
and second because there are bounds (called Fr\'echet bounds \cite{malevergne2006extreme}) within which $C(u,v)$ is constrained to live, 
and that compresse the difference between different copulas (the situation is even worse in higher dimensions). 
Estimating copula densities, on the other hand, is even more difficult than estimating univariate densities, especially in the tails. 
We therefore propose to focus on the diagonal of the copula, $C(p,p)$ and the anti-diagonal, $C(p,1-p)$, 
that capture part of the information contained in the full copula, and can be represented as 1-dimensional plots. 
Furthermore, in order to emphasise the difference with the case of independent variables, 
it is expedient to consider the following quantities:
\begin{subequations}
\begin{align}
	\frac{C(p,p)-p^2}{p(1\!-\!p)}&=\tUU(p)+\tLL(1-p)-1\\
	\frac{C(p,1\!-\!p)-p(1\!-\!p)}{p(1\!-\!p)}&=1-\tUL(p)-\tLU(1-p),
\end{align}
\end{subequations}
where the normalisation is chosen such that the tail correlations appear naturally. 
Note in particular that the diagonal quantity tends to the asymptotic tail dependence 
coefficients $\tUU^*$ ($\tLL^*$) when $p \to 1$ ($p \to 0$). 
Similarly, the anti-diagonal tends to $-\tUL^*$ as $p \to 1$ and to $-\tLU^*$ as $p \to 0$.

Another important reference point is the Gaussian copula $C_G(u,v)$, 
and we will consider below the normalised differences along the diagonal and the anti-diagonal:
\begin{equation}\label{eq:Delta}
	\Delta_d(p)=\frac{C(p,p)-C_G(p,p)}{p(1\!-\!p)}; \qquad \Delta_a(p)=\frac{C(p,1\!-\!p)-C_G(p,1\!-\!p)}{p(1\!-\!p)},
\end{equation}
which again tend to the asymptotic tail dependence coefficients in the limits $p \to 1$ ($p \to 0$), 
owing to the fact that these coefficients are all zero in the Gaussian case.

The centerpoint of the copula, $C(\tfrac{1}{2},\tfrac{1}{2})$, is particularly interesting: 
it is the probability that both variables are simultaneously below their respective median. 
For bivariate Gaussian variables, one can show that:
\be
C_G(\tfrac{1}{2},\tfrac{1}{2}) = \frac14 + \frac{1}{2 \pi} \arcsin \rho \equiv \frac14 \left(1 + \ros\right),
\ee
where $\ros$ is the sign correlation coefficient defined above. 
The trivial but remarkable fact is that the above expression holds for any elliptical models, 
that we define and discuss in the next section. 
This will provide a powerful test to discriminate between empirical data 
and a large class of copulas that have been proposed in the literature.

\section{Review of elliptical models}\label{sec:elliptical_model}

\subsection{General elliptical models}
Elliptical random variables $X_i$ can be simply generated by multiplying 
standardised Gaussian random variables $\epsilon_i$ with a common random 
(strictly positive) factor $\sigma$, drawn independently from the $\epsilon$'s 
from an arbitrary distribution $P(\sigma)$: \cite{embrechts02correlation}
\begin{equation}
	X_i=\mu_i+\sigma\cdot\epsilon_i\qquad{}i=1\ldots{}N
\end{equation}
where $\mu_i$ is a location parameter (set to zero in the following), 
and \mbox{$\boldsymbol\epsilon\sim\mathcal{N}(\mathbf{0},\Sigma)$}. 
Such models are called ``elliptical'' because the corresponding multivariate probability distribution function (pdf)
depends only on the rotationally invariant quadratic form $(\boldsymbol{x-\mu})^T \Sigma^{-1} (\boldsymbol{x-\mu})$; 
therefore quantile levels define ellipses in $N$ dimensions 
(see \cite{cambanis1981theory} for the construction and properties of elliptically contoured multivariate distributions). 

In a financial context, elliptical models are basically stochastic volatility models with arbitrary dynamics such that the 
ergodic distribution is $P(\sigma)$ (since we model only single-time distributions, the time ordering is irrelevant here).
The important assumption, however, is that the random amplitude factor $\sigma$ is \emph{equal for all stocks} (or all assets in 
a more general context). In other words, one assumes that the mechanisms leading to increased levels of volatility 
affect all individual stocks identically. There is a unique market volatility factor. Of course, this is a very restrictive 
assumption since one should expect {\it a priori} other, sector specific, sources of volatility. This, we believe, is the 
main reason for the discrepancies between elliptical models and the empirical data reported below.

Some of the above measures of dependence can be explicitely computed for elliptical models. Introducing the following ratios:
$f_d=\mathds{E}[\sigma^{2d}]/\mathds{E}[\sigma^d]^2$, one readily finds:
\begin{subequations}\label{eq:coeffs_ell}
\begin{align}
	\roc(\rho)&=\frac{f_2\cdot(1+2\rho^2)-1}{3 f_2-1}\\
	\roa(\rho)&=\frac{f_1\cdot D(\rho)-1}{\frac{\pi}{2}f_1-1}\\\label{eq:C05_ell}
	 C^*(\rho)&=\frac{1}{4}+\frac{1}{2\pi}\arcsin \rho = \frac{1}{4}\left(1+\ros\right),
\end{align}
\end{subequations}
where $C^* \equiv C(\tfrac{1}{2},\tfrac{1}{2})$.
Note that $f_2$ is related to the kurtosis of the $X$'s through the relation $\kappa=3(f_2-1)$. 
The calculation of the tail correlation coefficients depends on the specific form of $P(\sigma)$, 
for which several choices are possible. We will focus in the following on two of them, 
corresponding to the Student model and the log-normal model. 

\subsection{The Student ensemble} 

When the distribution of the square volatility is inverse Gamma, i.e. $P(u=\sigma^2) \propto \e^{-\frac{1}{u}}/u^{1+\nu/2}$,
the joint pdf of the returns turns out to have an explicit form \cite{demarta2005t,embrechts02correlation,malevergne2006extreme}:
\begin{equation}
   t_{\nu}(\boldsymbol{x})=\frac{1}{\sqrt{(\nu\pi)^N\det\Sigma}}\frac{\Gamma(\frac{\nu+N}{2})}
   {\Gamma(\frac{\nu}{2})}\left(1+\frac{\boldsymbol{x}^T\Sigma^{-1}\boldsymbol{x}}{\nu}\right)^{-\frac{\nu+N}{2}}
\end{equation}
This is the multivariate Student distribution with $\nu$ degrees of freedom for $N$ random variables with dispersion matrix $\Sigma$.
Clearly, the marginal distribution of $t_{\nu}(\boldsymbol{x})$
is itself a Student distribution with a tail exponent equal to $\nu$, which is well known to describe satisfactorily the univariate pdf of high-frequency returns
(from a few minutes to a day or so), with an exponent in the range $[3,5]$ (see e.g. Refs.~\cite{bouchaud2003theory,cont2001empirical,fuentes2009universal,plerou1999scaling}).
The multivariate Student model is therefore a rather natural choice; its corresponding copula defines the Student copula. 
For $N=2$, it is entirely parameterised by $\nu$ and the correlation coefficient $\rho=\Sigma_{12}/\sqrt{\Sigma_{11}\Sigma_{22}}$.

The moments of $P(\sigma)$ are easily computed and lead to the following expressions for the coefficients $f_d$:
\begin{subequations}
\begin{align}
  f_1&=\frac{2}{\nu-2}\left(\frac{\Gamma(\frac{\nu}{2})}{\Gamma(\frac{\nu-1}{2})}\right)^2&
  f_2&=\frac{\nu-2}{\nu-4}
\end{align}
\end{subequations}
when $\nu>2$ (resp. $\nu>4$).  Note that in the limit $\nu \to \infty$ at fixed $N$, the multivariate Student distribution boils down to a 
multivariate Gaussian. The shape of $\roa(\rho)$ and $\roc(\rho)$ for $\nu=5$ is given in Fig.~\ref{fig:dep}.

\begin{figure}[t!]
	\center
	\includegraphics[scale=0.4,angle=-90]{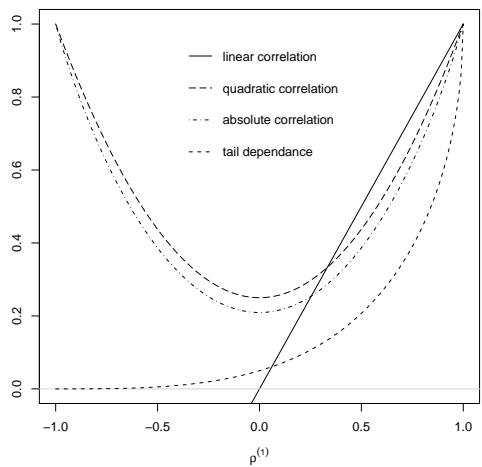}
	\caption{For all elliptical distributions, the different measures of dependence are functions of $\rho$.
	         Illustration for a pair with a bivariate Student distribution ($\nu=5$).}
	\label{fig:dep}
\end{figure}

One can explicitely compute the coefficient of tail dependence for Student variables,
which only depends on $\nu$ and on the linear correlation coefficient $\rho$. 
By symmetry, one has $\tUU(p;\nu,\rho)=\tLL(p;\nu,\rho)=\tUL(p;\nu,-\rho)=\tLU(p;\nu,-\rho)$. 
When $p = 1 - \epsilon$ with $\epsilon \to 0$, one finds:
\begin{equation}\label{eq:tau_puissance}
	\tUU(p;\nu,\rho)=\tUU^*(\nu,\rho) + \beta(\nu,\rho)\cdot \epsilon^{\frac{2}{\nu}} + \mathcal{O}(\epsilon^{\frac{4}{\nu}}),
\end{equation}
where $\tUU^*(\nu,\rho)$ and $\beta(\nu,\rho)$ are coefficients given in \ref{apx:tau_puissance}, 
see also Figs.~\ref{fig:dep}~and~\ref{fig:tau_stud}. The notable features of the above results are:
\begin{itemize}
\item The asymptotic tail dependence $\tUU^*(\nu,\rho)$ is strictly positive for all $\rho > -1$ and finite $\nu$, 
and tends to zero in the Gaussian limit $\nu \to \infty$ (see Fig.~\ref{fig:tau_stud1}). 
The intuitive interpretation is quite clear: large events are caused by large occasional bursts in volatility. 
Since the volatility is common to all assets, the fact that one return is positive extremely large is enough to infer that the volatility is itself large.
Therefore that there is a non-zero probability $\tUU^*(\nu,\rho)$ that another asset also has a large positive return 
(except if $\rho=-1$ since in that case the return can only be large and negative!). 
It is useful to note that  the value of $\tUU^*(\nu,\rho)$ does not depend on the explicit shape of $P(\sigma)$ 
provided the asymptotic tail of $P(\sigma)$ decays as $L(\sigma)/\sigma^{1+\nu}$, where $L(\sigma)$ is a slow function. 

\item The coefficient $\beta(\nu,\rho)$ is also positive, indicating that estimates of $\tUU^*(\nu,\rho)$ 
based on measures of $\tUU(p;\nu,\rho)$ at finite quantiles (e.g. $p=0.99$) are biased upwards. 
Note that the correction term is rapidly large because $\epsilon$ is raised to the power $2/\nu$. 
For example, when $\nu=4$ and $\rho=0.3$, $\beta\approx 0.263$ and one expects a first-order correction $0.026$ for $p=0.99$.
This is illustrated in Figs.~\ref{fig:tau_stud1},~\ref{fig:tau_stud2}.
The form of the correction term (in $\epsilon^{2/\nu}$) is again valid as soon as $P(\sigma)$ decays asymptotically as $L(\sigma)/\sigma^{1+\nu}$.

\item Not only is the correction large, but the accuracy of the first order expansion is very bad, 
since the ratio of the neglected term to the first correction is itself $\sim \epsilon^{2/\nu}$. 
The region where the first term is accurate is therefore exponentially small in $\nu$ --- see Fig.~\ref{fig:tau_stud2}. 
\end{itemize}

\begin{figure}[t!]
\center
    \subfigure[$\tUU$ vs. $\nu$ at several thresholds $p$ for bivariate Student variables with $\rho=0.3$. Note that $\tUU^* \to 0$ when $\nu \to \infty$, but rapidly 
    grows when $p \neq 1$.]{\label{fig:tau_stud1}\includegraphics[scale=0.4]{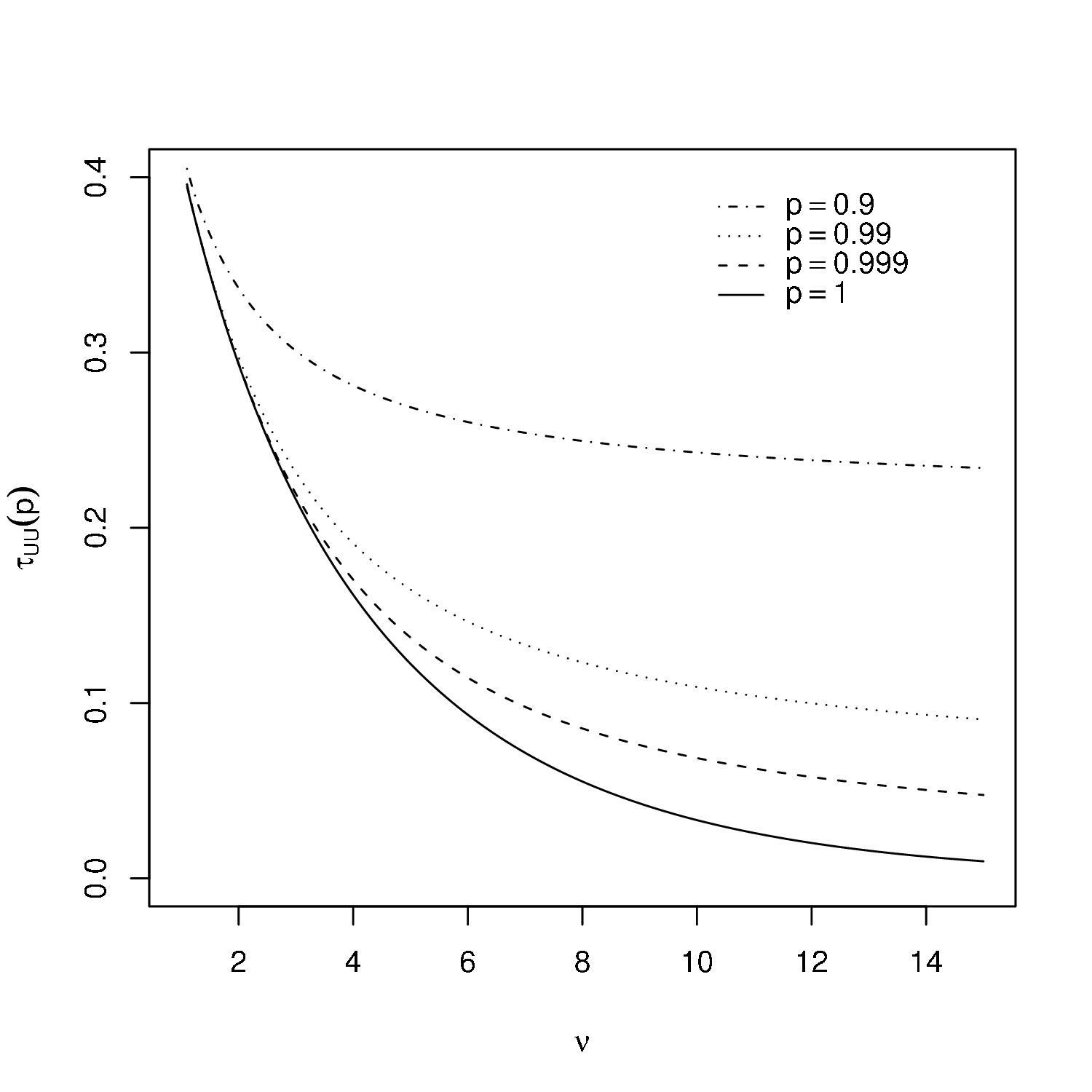}}\hfill
    \subfigure[$\tUU$ vs $p$ for bivariate Student variables with $\rho=0.3$ and $\nu=5$: exact curve (plain), first order power-law expansion (dashed) and simulated series with different lengths $T$ 
    (symbols)]{\label{fig:tau_stud2}\includegraphics[scale=0.4,trim=0 0 420 0,clip]{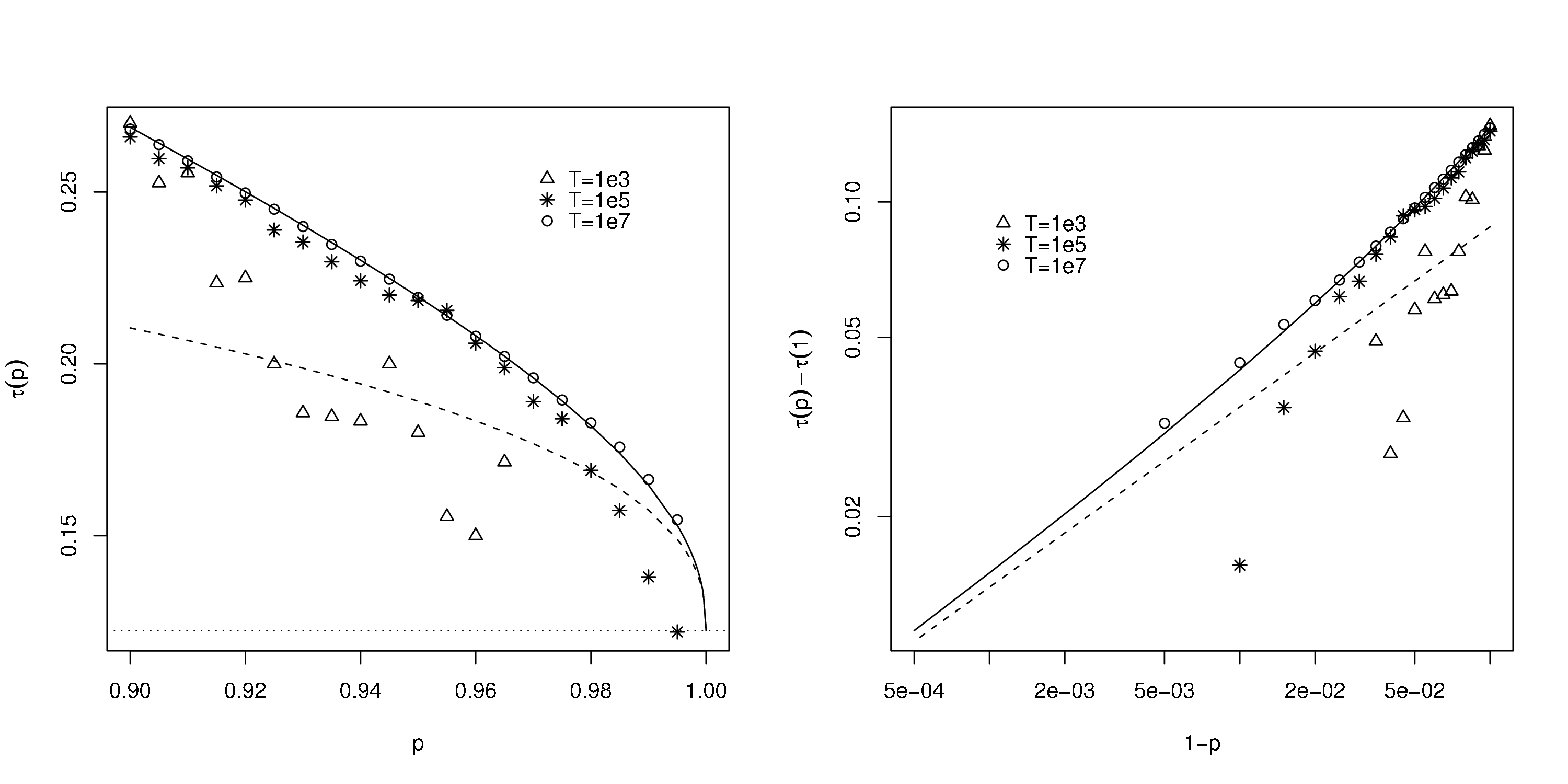}}
    \caption{}	
    \label{fig:tau_stud}
\end{figure}

Finally, we plot in Fig.~\ref{fig:dev_gauss_stud} the rescaled difference between the Student copula and the Gaussian copula 
both on the diagonal and on the anti-diagonal,
for several values of the linear correlation coefficient $\rho$ and for $\nu=5$. 
One notices that the difference is zero for $p=\tfrac{1}{2}$, as expected from the expression of $C^*(\rho)$ for general elliptical models. 
Away from $p=\tfrac{1}{2}$ on the diagonal, the rescaled difference has a positive convexity and non-zero limits when $p \to 0$ and $p \to 1$,
corresponding to $\tLL^*$ and $\tUU^*$.

\begin{figure}[ht]
\center
    \includegraphics[scale=0.4]{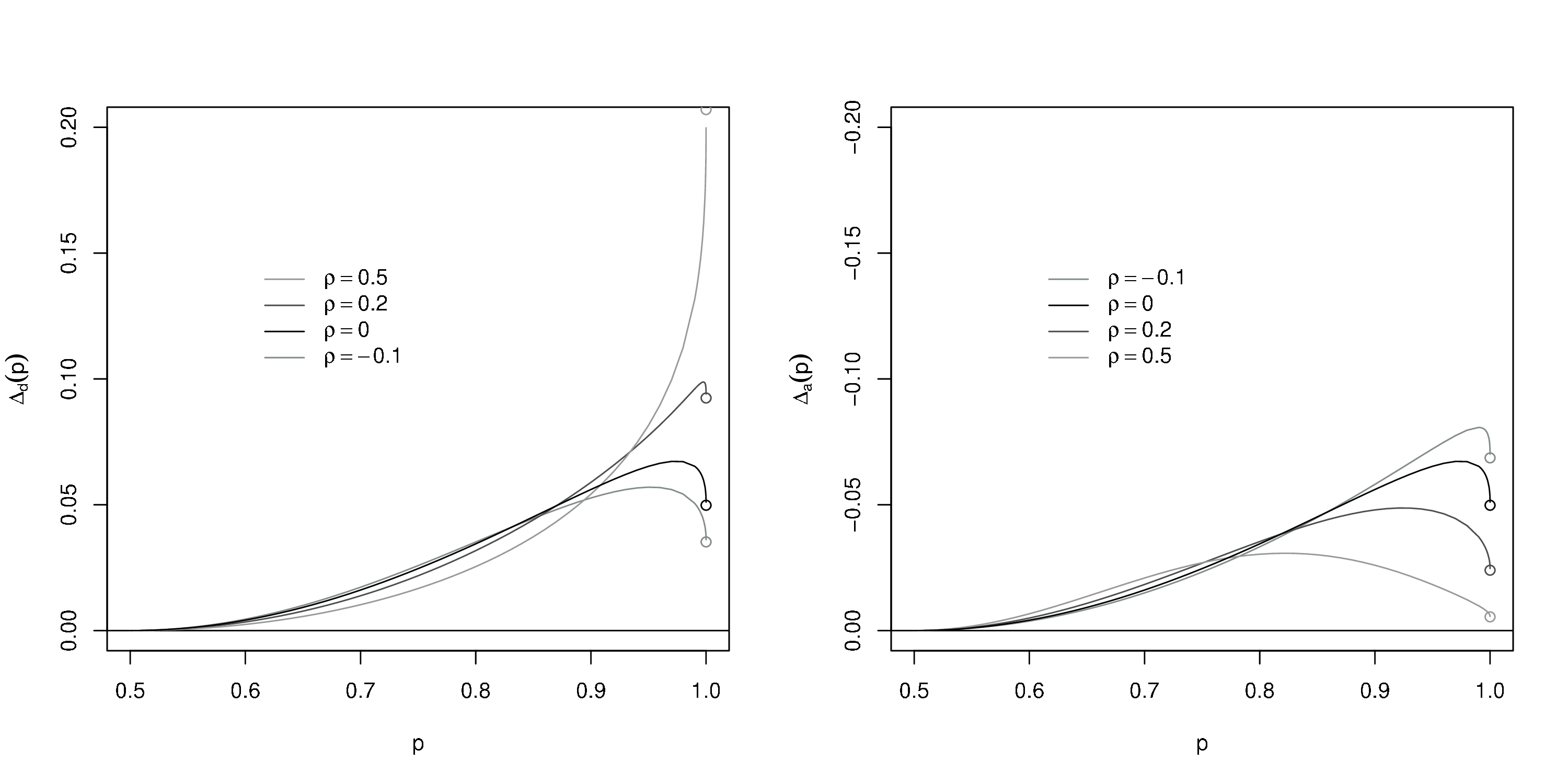}
    \caption{Diagonal and anti-diagonal of the Student copula: 
    the quantities $\Delta_d(p)$ and $\Delta_a(p)$ defined in Eq.~\eqref{eq:Delta} are plotted versus $p$ for several values of $\rho$ and fixed $\nu=5$.
    (The curves are identical in the range $p\in[0,0.5]$ due to the symmetry $p\leftrightarrow 1\!-\!p$).}	
    \label{fig:dev_gauss_stud}
\end{figure}

\subsection{The log-normal ensemble} 

If we now choose $\sigma=\sigma_0 \e^\xi$ with $\xi \sim \mathcal{N}(0,s)$ (as suggested by advocates of multifractal models \cite{calvetfisher,lux,bacrymuzy}), the resulting multivariate 
model defines the log-normal ensemble and the log-normal copula. The factors $f_d$ are immediately found to be: $f_d=\e^{d^2s^2}$ with no restrictions on $d$.
The Gaussian case now corresponds to $s \equiv 0$.

Although the inverse Gamma and the lognormal distributions are very different, it is well known that the tail region of the log-normal is in practice
very hard to distinguish from a power-law. In fact, one may write:
\begin{equation}
	\widehat \sigma^{-1} \e^{-\frac{\ln^2 \widehat \sigma}{2 s^2}} =  
	\widehat \sigma^{-(1+\nu_{\text{eff}})} \qquad \nu_{\text{eff}} = \frac{\ln \widehat \sigma}{2 s^2},
\end{equation}
where we have introduced $\widehat \sigma = \sigma/\sigma_0$. The above equation means that for large $\widehat \sigma$ a lognormal is like a power-law with a 
slowly varying exponent. If the tail region corresponds to $\widehat \sigma \in [4,8]$ (say), the effective exponent $\nu$ lies in the range $[0.7/s^2,1.05/s^2]$.
Another way to obtain an approximate dictionary between the Student model and the lognormal model is to fix $\nu_{\text{eff}}(s)$ such that the coefficient
$f_2$ (say) of the two models coincide. Numerically, this leads to $\nu_{\text{eff}}(s) \approx 2 + 0.5/s^2$, leading to $s \approx 0.4$ for $\nu=5$. 
In any case, our main point here is that from an empirical point of view, one does not expect striking differences between a Student model and a log-normal model --- even 
though from a strict mathematical point of view, the models are very different. In particular, the asymptotic tail dependence coefficients $\tUU^*$ or $\tLL^*$ are always zero
for the log-normal model (but $\tLL(p)$ or $\tUU(p)$ converge exceedingly slowly towards zero as $p \to 1$).

\subsection{Pseudo-elliptical generalisation}\label{ssec:pseudo-ell}

In the previous elliptical description of the returns, all stocks are subject to 
the exact same stochastic volatility, what leads to non-linear dependences like tail effects and residual higher-order correlations even for $\rho=0$ (see Fig.~\ref{fig:dep}).
In order to be able to fine-tune somewhat this dependence due to the common volatility, a simple generalisation is to let each stock be 
influenced by an idiosyncratic volatility, thus allowing for a more subtle structure of dependence. More specifically, we write\footnote{{
In vectorial form, we collect the individual (yet dependent) stochastic volatilities $\sigma_i$ on the diagonal of a matrix $D$, 
and $X=D\epsilon$ can be decomposed as $X=RDAU$
where $U$ is a random vector uniformly distributed on the unit hypersphere, the radial component $R$ is a chi-2 random variable independent of $U$ with 
$\operatorname{rank}(A)$ degrees of freedom, and $A$ is a matrix with appropriate dimensions such that $AA^{\dagger}=\Sigma$.
This description can be contrasted with the one proposed in Ref.~\cite{kring2009multi} under the term ``Multi-tail Generalized Elliptical Distribution'':
$X=R(\vec{u})AU$ with $R$ now depending on the unit vector $\vec{u}=AU/||AU||$.
In other words, the description in $\eqref{eq:pseudo_ell_coeff}$ provides a different radial amplitude for each component, whereas \cite{kring2009multi}
characterises a direction-dependent radial part identical for every component.
The latter allows for a richer phenomenology than the former, but lacks financial intuition.}}: 
\begin{equation}\label{eq:pseudo_ell_coeff}
	X_i=\sigma_i\cdot\epsilon_i\qquad{}i=1\ldots{}N
\end{equation}
where the Gaussian residuals $\epsilon_i$ have the same joint distribution as before, and are independent of the $\sigma_i$, 
%
%
but we now generalise the definition of the ratios $f_d$:
\begin{equation}\label{eq:fd_general}
f_d^{\scriptscriptstyle (i,j)}=\frac{\mathds{E}[\sigma_i^d\sigma_j^d]}{\mathds{E}[\sigma_i^d]\mathds{E}[\sigma_j^d]}
\end{equation}
which describe the joint distribution of the volatilities. As an explicit example, we consider the natural generalisation of the log-normal 
model and write $\sigma_i = \sigma_{0i} \e^{\xi_i}$, with\footnote{A further generalisation that allows for stock dependent ``vol of vol'' $s_i$ is also
possible.} $\xi_i\stackrel{\text{\tiny id}}{\sim}\mathcal{N}(0,s)$, and some correlation structure of the $\xi$'s: $\mathds{E}[\xi_i \xi_j]=s^2 c_{ij}$. One then finds:
\begin{equation}\label{eq:fd_lognorm}
f_d(c)=\e^{d^2 s^2 c}.
\end{equation}

Within this setting, the generalisation of coefficients \eqref{eq:coeffs_ell} can be straightforwardly calculated. 
Denoting by $r$ the correlation coefficient of $\epsilon_i$ and $\epsilon_j$ 
(now different from $\rho$) we find:
\begin{subequations}\label{eq:coeffs_pseudo_ell}
\begin{align}
	\rho(r,c)&=\frac{f_1(c)}{f_1(1)}\cdot{}r\\
	\roc(r,c)&=\frac{f_2(c)(1+2r^2)-1}{3f_2(1)-1}\\
	\roa(r,c)&=\frac{f_1(c)D(r)-1}{\frac{\pi}{2}f_1(1)-1}\\
	C^*(r,c)&=\frac{1}{4}+\frac{1}{2\pi}\arcsin r\quad\forall{}c
\end{align}
\end{subequations}
These formulas are straightforwardly generalised for arbitrary description of the volatilities, using the coefficients $f_d$ defined in
\eqref{eq:fd_general} instead of the explicit expressions given by \eqref{eq:fd_lognorm}.
When $c$ is fixed (e.g.\ for the elliptical case $c=1$), $\rho$ and $r$ are proportional, and 
all measures of non-linear dependences can be expressed as a function of $\rho$. But this ceases to be true as soon
as there exists some non trivial structure $c$ in the volatilities. In that case, $c$ and $r$ are ``hidden'' underlying variables, that
can only be reconstructed from the knowledge of $\rho$, $\roc$, $C^*$, assuming of course that the model is accurate.

Notice that the result on $C(\frac12,\frac12)$ is totally independent of the structure of the volatilities\footnote{ 
This property holds even when $\sigma_i$ depends on the sign of $\epsilon_i$, which might be useful to model the leverage
effect that leads to some asymmetry between positive and negative tails, as the data suggests.}.
Indeed, what is relevant for the copula at the central point is not the amplitude of the returns, 
but rather the number of $+/-$ events, 
which is unaffected by any multiplicative scale factor as long as the median of the univariate marginals is nil. 
An important consequence of this result is that for all elliptical or pseudo-elliptical model,
$\rho=0$ implies that \mbox{$C^*(\rho=0)=\frac14$}. 
We now turn to empirical data to test the above predictions of elliptical models, in particular this last one.

\subsection{A word on Archimedean copulas}

In the universe of all possible copulas, a particular family has become increasingly popular in finance: that of 
``Archimedean copulas''. These copulas are defined, for bivariate random vectors, as follows \cite{wuvaldez}:  
  \be
  C_A(u,v) \equiv  \phi^{-1} \left[\phi(u) + \phi(v)\right],
  \ee
  where $\phi(u): [0,1] \to [0,1]$ is a function such that $\phi(1)=0$ and $\phi^{-1}$ is decreasing and completely monotone. 
  For example, Frank copulas are such that $\phi_F(u) = \hbox{$\ln [\e^\theta - 1]$} - \hbox{$\ln [\e^{\theta u} - 1]$}$ where $\theta$ is a real parameter, or
  Gumbel copulas, such that $\phi_G(u) = (-\ln u)^\theta$, $0 < \theta \leq 1$. The asymptotic coefficient of tail dependence are all zero for
  Frank copulas, whereas $\tUU^*=2-2^\theta$ (and all other zero) for the Gumbel copulas. The case of general multivariate copulas is obtained 
  as natural generalisation of the above definition.
  
  One can of course attempt to fit empirical copulas with a specific Archimedean copula. By playing enough with the function $\phi$, it is obvious that
  one will eventually reach a reasonable quality of fit. What we take issue with is the complete lack of intuitive interpretation or plausible 
  mechanisms to justify why the returns of two correlated assets should be described with Archimedean copulas. This is particularly clear after 
  recalling how two Archimedean random variables are generated: first, take two $\mathcal{U}[0,1]$ random variables $s,w$. Set $w'=K^{-1}(w)$ with $K(x)=x-\phi(x)/\phi'(x)$.
  Now, set:
  \be
  u = \phi^{-1} \left[s\phi(w')\right]; \qquad v = \phi^{-1} \left[(1-s)\phi(w')\right],
  \ee
  and finally write $X_1= {\mathcal P}_{\negmedspace{\scriptscriptstyle <},1}^{-1}(u)$ and $X_2= {\mathcal P}_{\negmedspace{\scriptscriptstyle <},2}^{-1}(v)$ 
  to obtain the two Archimedean returns with the required marginals ${\mathcal P}_1$ and ${\mathcal P}_2$ \cite{genest1993statistical,wuvaldez}.
  Unless one finds a natural economic or microstructural interpretation for such a labyrinthine construction, 
  we content that such models should be discarded {\it a priori}, for lack of plausibility.
  In the same spirit, one should really wonder why the financial industry has put so much faith in Gaussian copulas
  models to describe correlation between {\it positive} default times, that were then used to price CDO's and other credit derivatives. 
  We strongly advocate the need to build models bottom-up: mechanisms should come before any mathematical representation
  (see also \cite{mikosch2006copulas} and references therein for a critical view on the use of copulas).

\section{Empirical study of the dependence of stock pairs}\label{sec:empirics}

\subsection{Methodology: what we do and what we do not do}

As the title stresses, the object of this study is to dismiss the elliptical copula as description
of the multivariate dependence structure of stock returns.
Concretely, the null-hypothesis ``$\text{H}_0$: the joint distribution is elliptical''
admits as corollary ``all pairwise bivariate marginal copulas are elliptical and differ only by their linear correlation coefficient ''.
In other words, all pairs with the same linear correlation $\rho$ are supposed to have identical values of non-linear dependence measures, 
and this value is predicted by the elliptical model. Focusing the empirical study on the pairwise measures of dependence, we will reject $\text{H}_0$ by
showing that their average value over all pairs with a given $\rho$ is different from the value predicted by elliptical models.

Our methodology differs from usual hypothesis testing using statistical tools and goodness of fit tests, as can be encountered for example in \cite{malevergne2003testing} 
for testing the Gaussian copula hypothesis on financial assets. Indeed, the results of such tests are often not valid because financial time series are \emph{persistent} (although almost not linearly autocorrelated),
so that successive realisations cannot be seen as independent draws of an underlying distribution, see the recent discussion on this issue in \cite{chicheportiche2011goodness}.

\subsection{Data set and time dependence}

The dataset we considered is composed of daily returns of 1500 stocks labeled in USD in the period 1995--2009 (15 full years). 
We have cut the full period in three subperiods (1995--1999, 2000--2004 and 2005--2009), and also the stock universe into three
pools (large caps, mid-caps and small caps). We have furthermore extended our analysis to the Japanese stock markets. Qualitatively,
our main conclusions are robust and do not depend neither on the period, nor on the capitalisation or the market. Some results, on the
other hand, do depend on the time period, but we never found any strong dependence on the capitalisation. 
All measures of dependence are calculated pairwise with non-parametric estimators, and
using all trading dates shared by both equities in the pair.

We first show the evolution of the linear correlation coefficients (Fig.~\ref{fig:rho_evol}) and the upper and lower tail dependence coefficients for $p=0.95$ (Fig.~\ref{fig:tau_evol}) as a function of
time for the large-cap stocks. One notes that (a) the average linear correlation ${\overline \rho}(t)$ fluctuates quite a bit, from around $0.1$ in the mid-nineties to around $0.5$ during 
the 2008 crisis; (b) the distribution of correlation coefficients shifts nearly rigidly around the moving average value ${\overline \rho}(t)$; (c) there is a marked, time 
dependent asymmetry between the average upper and lower tail dependence coefficients; (d) overall, the tail dependence tracks the behaviour of ${\overline \rho}(t)$. More
precisely, we show the time behaviour of the tail dependence assuming either a Gaussian underlying copula or a Student ($\nu=5$) copula with the same average correlation ${\overline \rho}(t)$.
Note that the former model works quite well when ${\overline \rho}(t)$ is small, whereas the Student model fares better when ${\overline \rho}(t)$ is large. We will repeatedly come back to this point below.

\begin{figure}[p]
\center
\includegraphics[scale=0.37,clip]{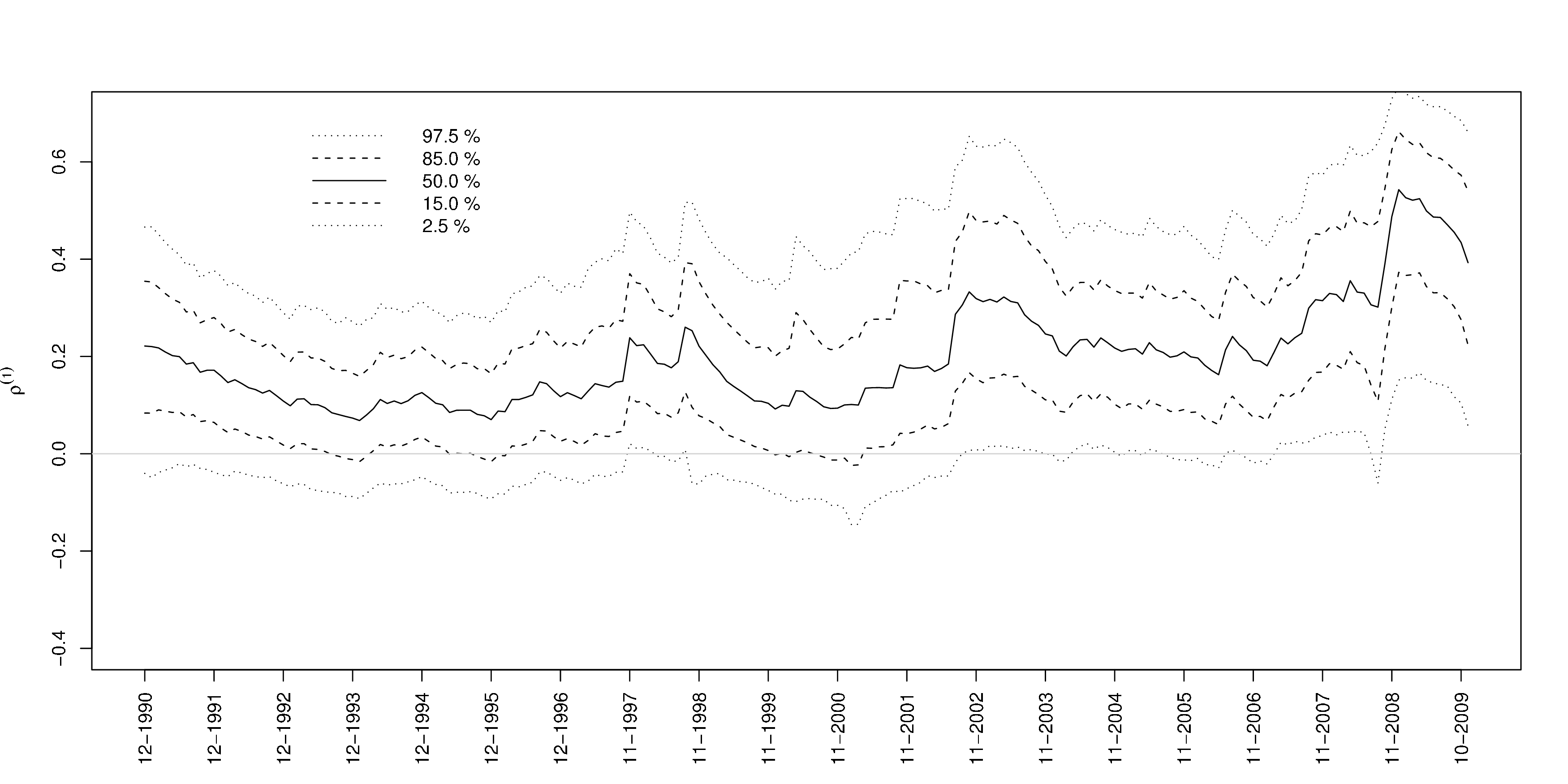}
\caption{Time evolution of the coefficient of linear correlation,
         computed for the stocks in the S\&P500 index, with an exponentially moving average of 125 days
         from January 1990 to December 2009.
         Five quantiles of the $\rho$ distribution are plotted according to the legend, showing that this distribution moves quasi-rigidly around its median value.}
\label{fig:rho_evol}
%
\vfill
\includegraphics[scale=0.37,clip]{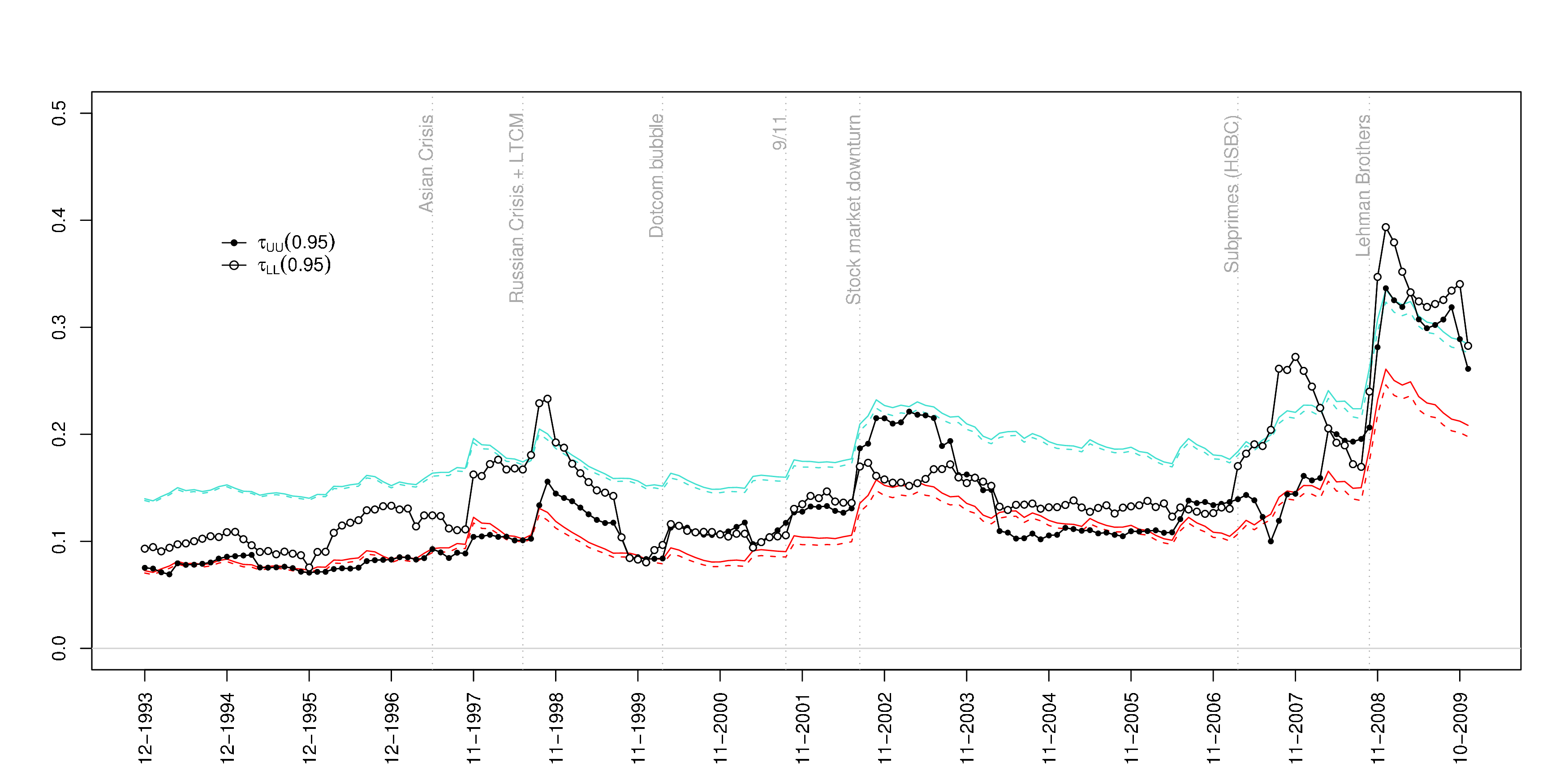}
\caption{Time evolution of the coefficient of tail dependence at $p=0.95$,
         computed for the stocks in the S\&P500 index. 
         A sliding flat window of 250 days moves monthly (25 days) from January 1990 to December 2009.
         We depict in turquoise and red the predictions for Student pair with $\nu=5$ and $\nu=\infty$ respectively, 
         based on the evolution of the mean value $\bar\rho(t)$ of the cross-sectional linear correlation (dashed) or averaged over the full instantaneous distribution of $\rho(t)$ (plain), 
         with little difference.}
\label{fig:tau_evol}
\end{figure}

We computed the quadratic $\roc$ and absolute $\roa$ correlation coefficients for each pair in the pool,
as well as the whole rescaled diagonal and anti-diagonal of the empirical copulas 
(this includes the coefficients of tail dependence\footnote{
The empirical relative bias for $p=0.95$ is $\lesssim{}1\%$ even for series with only $10^3$ points --- typically daily returns over 4 years.},
see Sect.~\ref{ssec:copulas} above).
We then average these observables over all pairs with a given linear coefficient $\rho$, within bins of varying width in order to take account of the frequency of observations in each bin. 
We show in Fig.~\ref{fig:all_vs_rho} $\roa$ and $\tUU(0.95), \tLL(0.95)$ as a function of $\rho$ for all stocks in the period 2000--2004, 
together with the prediction of the Student copula model for various values of $\nu$, 
including the Gaussian case $\nu = \infty$. 
For both quantities, we see that the empirical curves systematically cross the Student predictions, 
looking more Gaussian for small $\rho$'s and compatible with Student $\nu \approx 6$ for large $\rho$'s, 
echoing the effect noticed in Fig.~\ref{fig:tau_evol} above. 
The same effect would appear if one compared with the log-normal copula: 
elliptical models imply a residual dependence due to the common volatility, 
\emph{even when the linear correlation goes to zero}. 
This property is not observed empirically, since we rather see that higher-order correlations almost vanish together with $\rho$. 
The assumption that a common volatility factor affects all stocks is therefore too strong,
although it seems to make sense for pairs of stocks with a sufficiently large linear correlation. 
This result is in fact quite intuitive and we will expand on this idea in Sect.~\ref{ssec:pseudo_ell_log} below. 

The above discrepancies with the predictions of elliptical models are qualitatively the same for all periods, market caps and is also found for the Japan data set.
Besides, elliptical models predict a symmetry between upper and lower tails, whereas the data suggest that the tail dependence is asymmetric. Although most of 
the time the lower tail dependence is stronger, there are periods (such as 2002) when the asymmetry is reversed.

\begin{figure}[t!h]\center
    \subfigure[\mbox{$\roa$       vs $\rho$}]{\label{fig:abscor_rho_data}\includegraphics[scale=0.4]{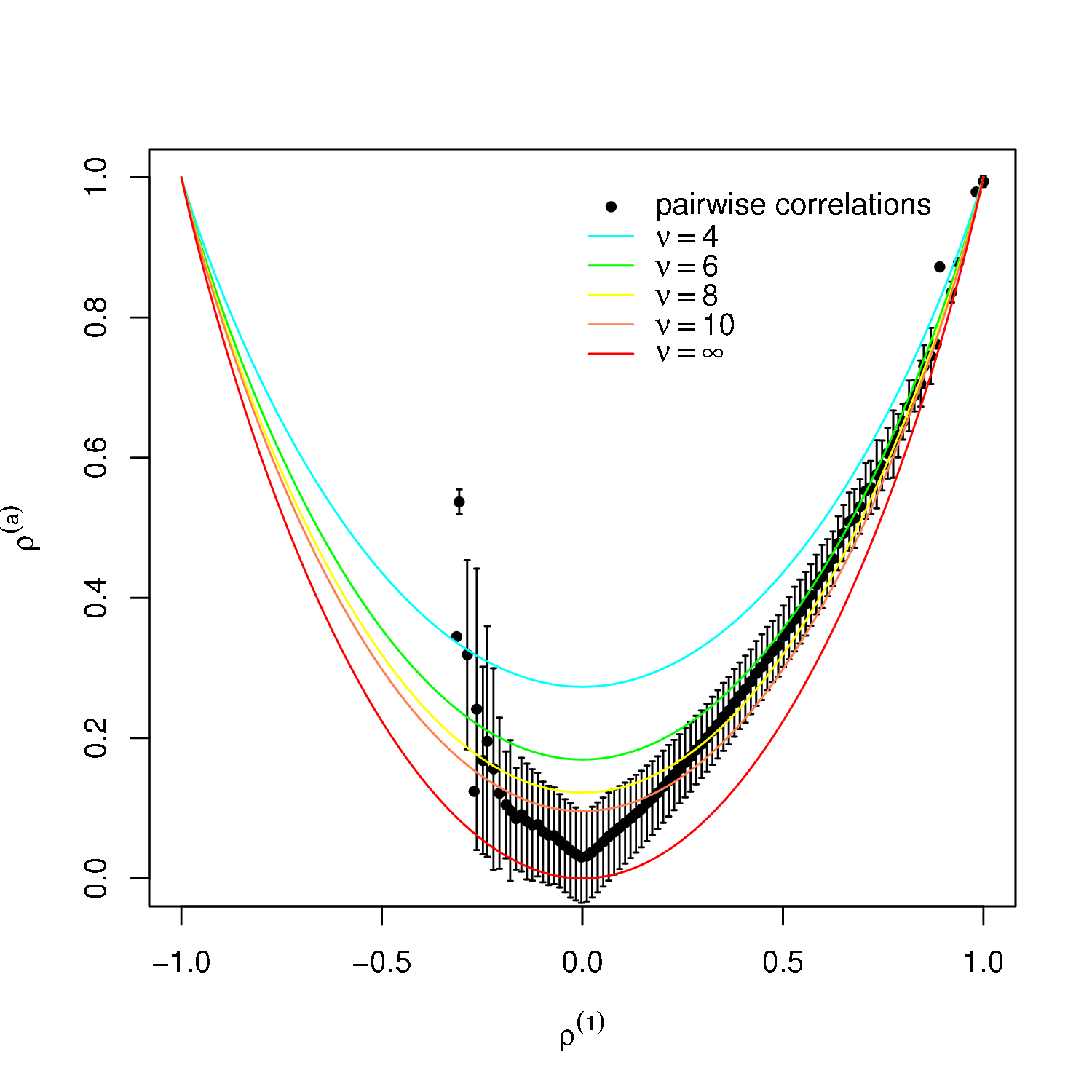}}\hfill
    \subfigure[\mbox{$\tau(0.95)$ vs $\rho$}]{\label{fig:tau_rho_data   }\includegraphics[scale=0.4]{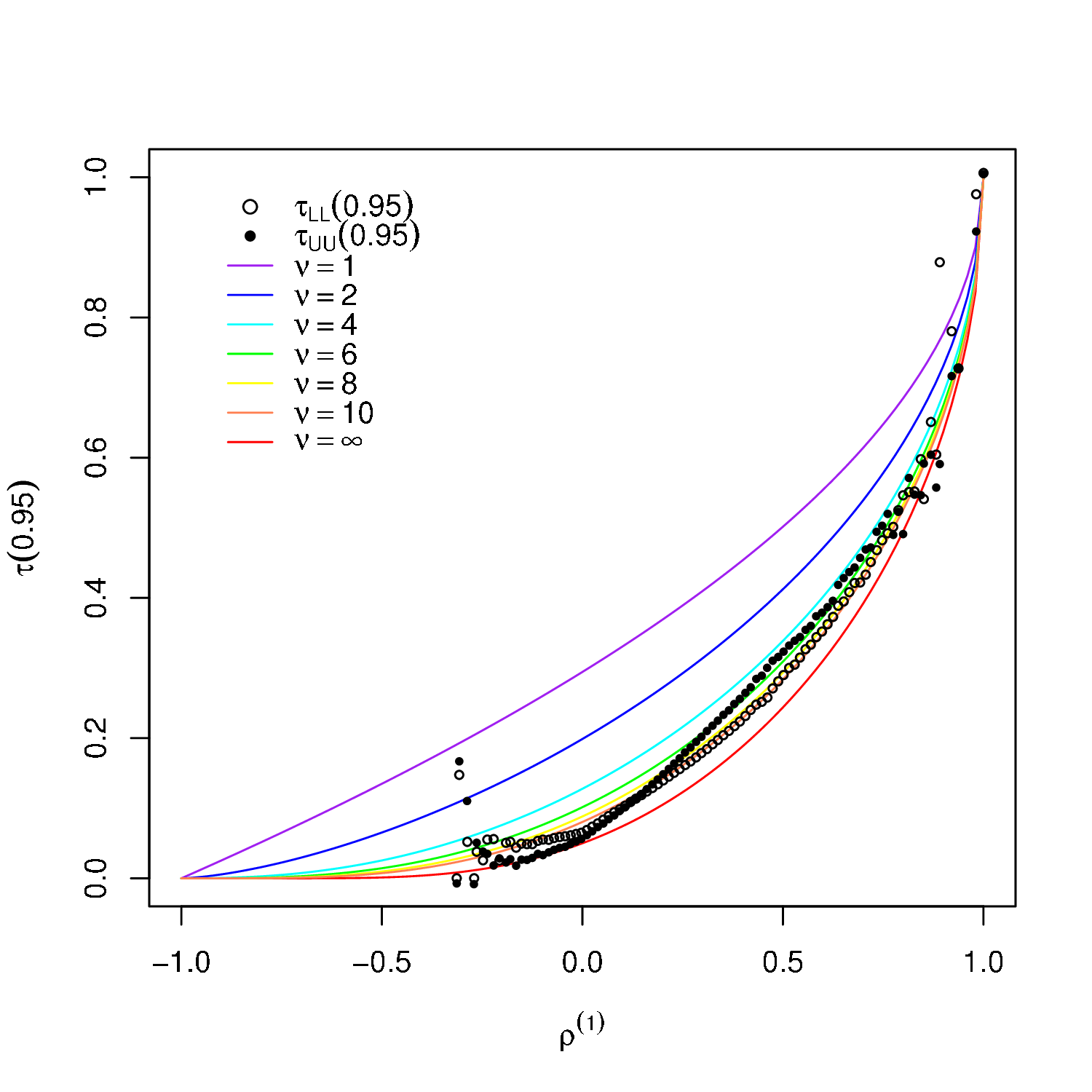}}
    \caption{Empirical (2000--2004) and elliptical. The fact that the empirical points ``cross'' the theoretical lines
  corresponding to elliptical models reveals that there is a dependence structure in the volatility, more complex
  than what is suggested by a single common stochastic scale factor.}	
  \label{fig:all_vs_rho}
\end{figure}

\subsection{Copula diagonals}

As argued in Sect.~\ref{ssec:copulas} above, it is convenient to visualise the rescaled difference $\Delta_{d,a}(p)$ 
between the empirical copula and the Gaussian copula, along the diagonal $(p,p)$ and the anti-diagonal $(p,1-p)$, see Eqs.~\eqref{eq:Delta}. 
The central point $p=\frac12$ is special; since $C_G(\frac12,\frac12)=\frac14$, $\Delta_{d,a}(\frac12)$
are both equal to $4 C^* - 1$, that takes a universal value for all elliptical models, given by Eq.~\eqref{eq:C05_ell}. 

We show in Fig.~\ref{fig:cop_dev} the diagonal and anti-diagonal copulas for all pairs of stocks in the period 2000-2004, 
and for various values of $\rho$. 
We also show the prediction of the Student $\nu=5$ model and of a Frank copula model with Student $\nu=5$ marginals and the appropriate value of $\rho$. 
What is very striking in the data is that $\Delta_d(p)$ is concave for small $\rho$'s and becomes convex for large $\rho$'s, 
whereas the Student copula diagonal is {\it always convex}. 
This trend is observed for all periods, and all caps, and is qualitatively very similar in Japan as well for all periods between 1991 and 2009. 
We again find that the Student copula is a reasonable representation of the data only for large enough $\rho$. 
The Frank copula is always a very bad approximation -- see the wrong curvature along the diagonal and the inaccurate behaviour along the anti-diagonal.

Let us now turn to the central point of the copula, $C^*=C(\frac12,\frac12)$. 
We plot in Fig.~\ref{fig:C05} the quantity $-\cos\left(2\pi C^*\right)-\rho$ as a function of $\rho$, 
which should be zero for all elliptical models, according to Eq.~\eqref{eq:C05_ell}. 
The data here include the 284 equities constantly member of the S\&P500 index in the period {2000--2009}. 

We again find a clear systematic discrepancy, that becomes stronger for smaller $\rho$'s: 
the empirical value of $C^*$ is too large compared with the elliptical prediction. 
In particular, for stocks with zero linear correlation ($\rho = 0$), we find $C^* > \frac14$, 
i.e. even when two stocks are uncorrelated, the probability that they move in the same direction\footnote{
A more correct statement is that both stocks have returns below their median with probability larger than $\frac14$. 
However, the median of the distributions are very close to zero, justifying our slight abuse of language.} is larger than $\frac14$. 
The bias shown in Fig.~\ref{fig:C05} is again found for all periods and all market caps, and for Japanese stocks as well. 
Only the amplitude of the effect is seen to change across the different data sets.

Statistical errors in each bin are difficult to estimate since pairs containing the same asset are mechanically correlated. 
In order to ascertain the significance of the previous finding, we have compared the empirical value of $C^*(\rho)$ with the result of a numerical
simulation, where we generate time series of Student ($\nu=5$) returns using the empirical correlation matrix. 
In this case, the expected result is that all pairs with equal correlation have the same bivariate copula and thus the same $C^*$ 
so that the dispersion of the results gives an estimate of measurement noise. We find that, as expected, $C^*(\rho)$ is compatible with Eq.~\eqref{eq:C05_ell},
at variance with empirical results.
We also find that the dispersion of the empirical points is significantly larger than that of the simulated elliptical pairs with identical linear correlations,
suggesting that all pairs {\it cannot be described by the same bivariate copula} (and definitively not an elliptical copula, as argued above).

All these observations, and in particular the last one, clearly indicate that Student copulas, or any elliptical copulas, are inadequate to represent the full 
dependence structure in stock markets.
Because this class of copulas has a transparent interpretation, we in fact know why this is the case: assuming a common random volatility 
factor for all stocks is oversimplified. This hypothesis is indeed only plausible for sufficiently correlated stocks, in agreement with the set of observations we made above. As a
first step to relax this hypothesis, we now turn to the pseudo-elliptical log-normal model, that allows further insights but still has unrecoverable failures.

\begin{figure}[tp]
	\center
  	\subfigure[]{\label{fig:C05a}\includegraphics[scale=0.25,clip]{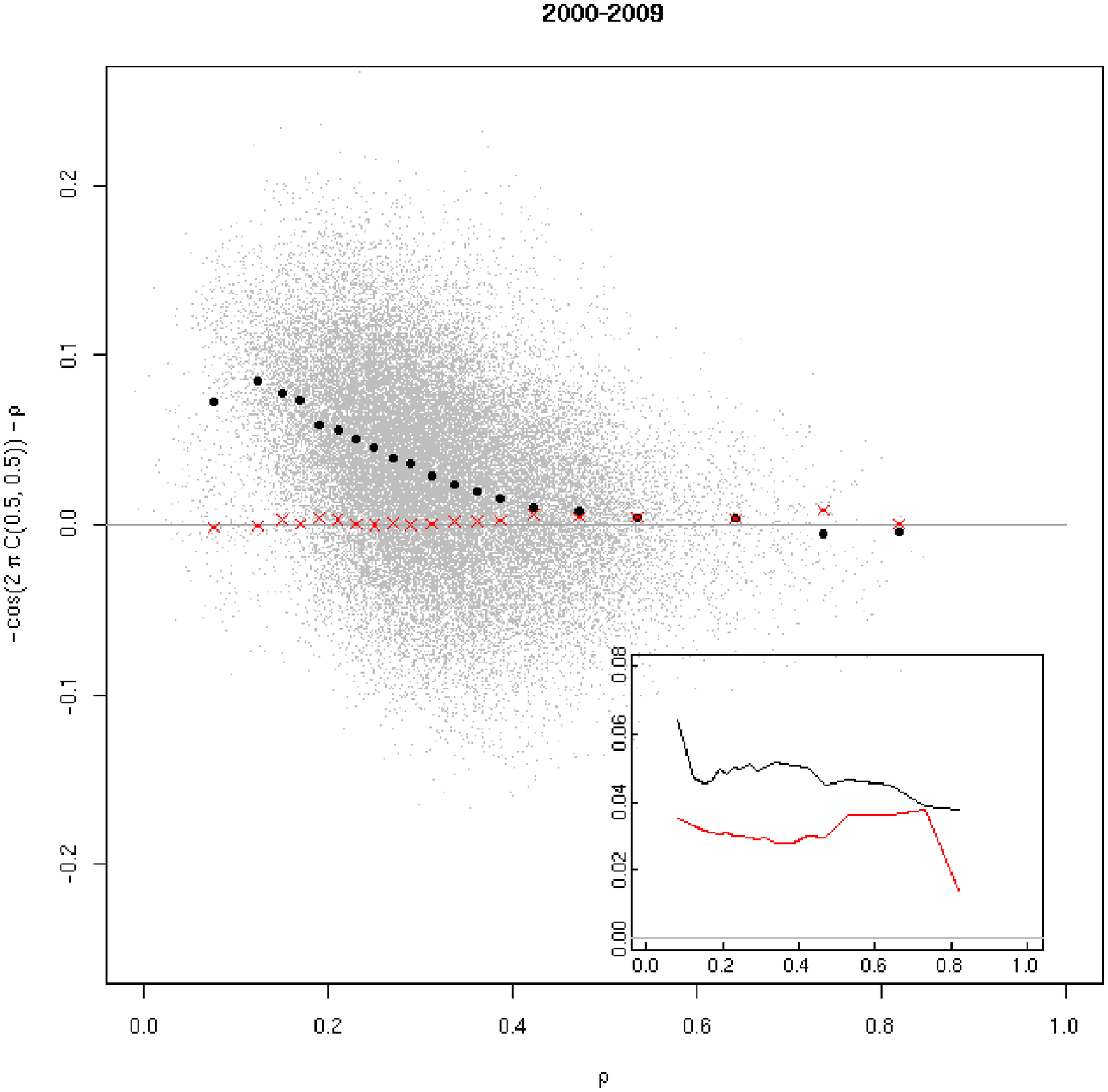}}\hfill
  	\subfigure[]{\label{fig:C05b}\includegraphics[scale=0.25,clip]{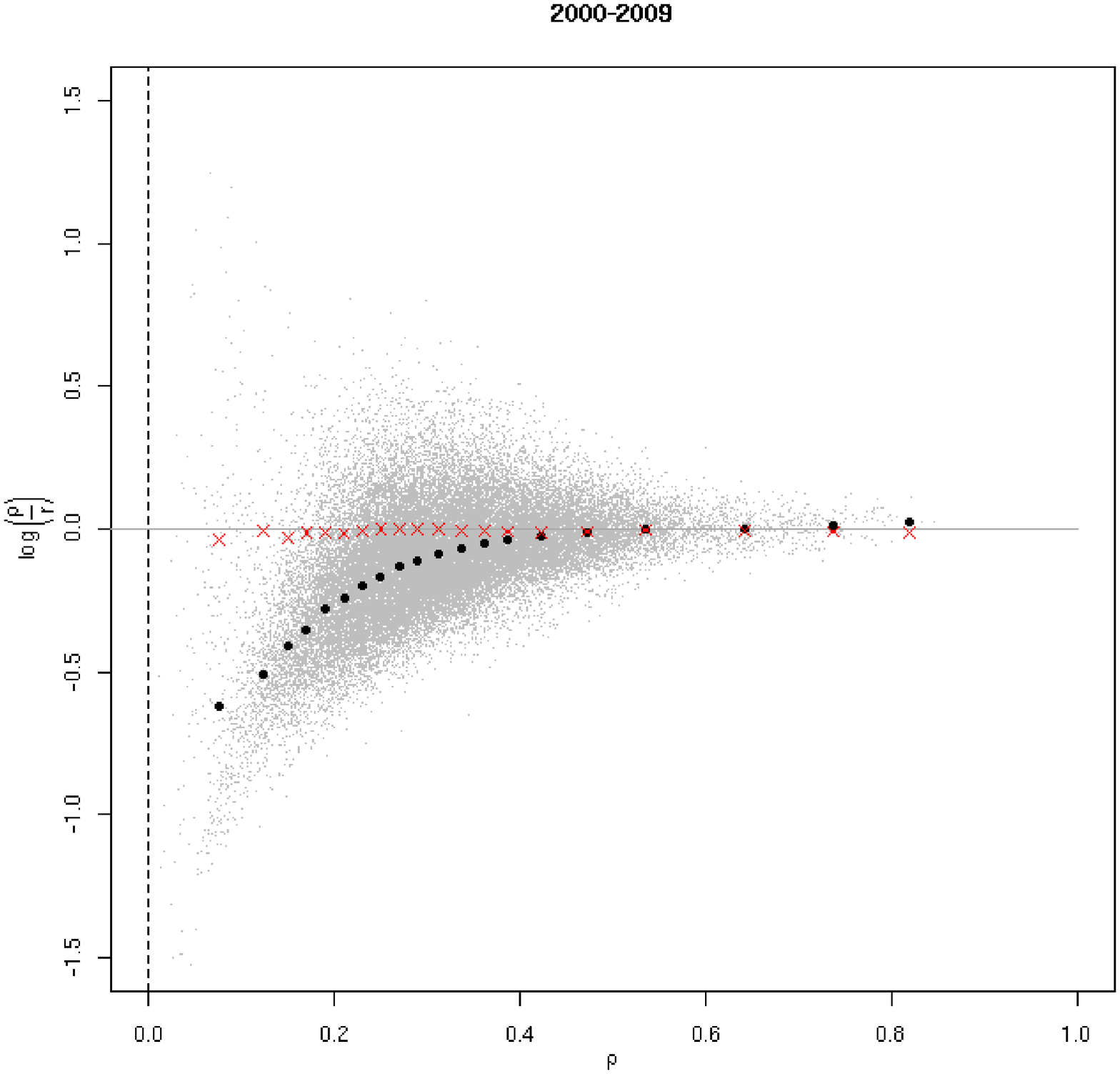}}
	\caption{{\textbf{Left:} $-(\cos\left(2\pi C^*\right)+\rho)$ versus the coefficient $\rho$ of linear correlation.
	The grey cloud is a scatter plot of all pairs of stocks.
	Black dots correspond to the average value of empirical measurements within a correlation bin,
	and red crosses correspond to the numerically generated elliptical (Student, $\nu=5$) data, compatible with the 
	expected value zero. 
	Statistical error are not shown, but their amplitudes is of the order of the fluctuations of the red crosses, 
	so that empirics and elliptical prediction don't overlap for $\rho\lesssim 0.4$.
	This representation allows to visualise very clearly the systematic discrepancies for small values of $\rho$, beyond the average value.
	Inset: the 1\,s.d. dispersion of the scattered points inside the bins. 
	\textbf{Right:} 
	\mbox{$z=\ln \left(\rho/|\cos\left(2\pi C^*\right)|\right)$} as a function of $\rho$, with the same data and the same conventions. 
	This shows more	directly that the volatility correlations of weakly correlated stocks is overestimated by elliptical models, whereas strongly 
	correlated stocks are compatible with the elliptical assumption of a common volatility factor.}}
	\label{fig:C05}
\end{figure}

\begin{figure}[!p!]
  \center
  \includegraphics[scale=0.4,clip]{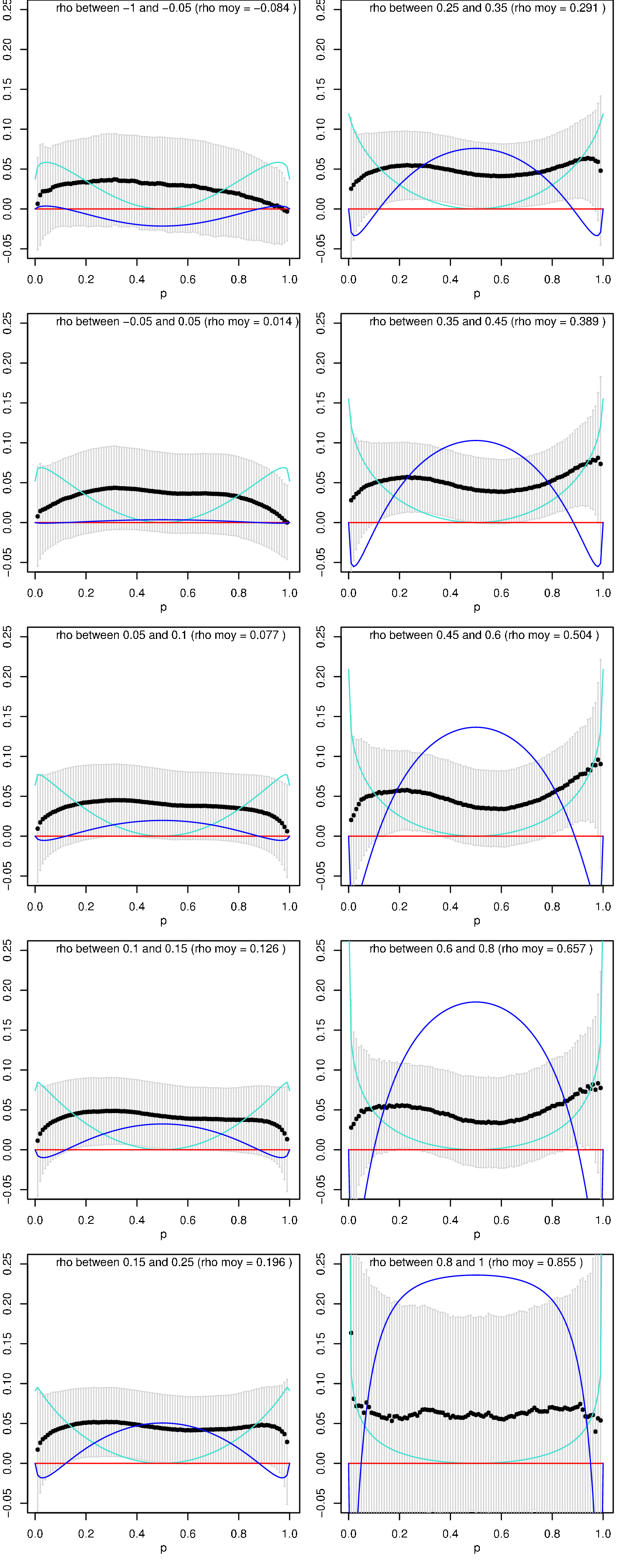}\hfill
  \includegraphics[scale=0.4,clip]{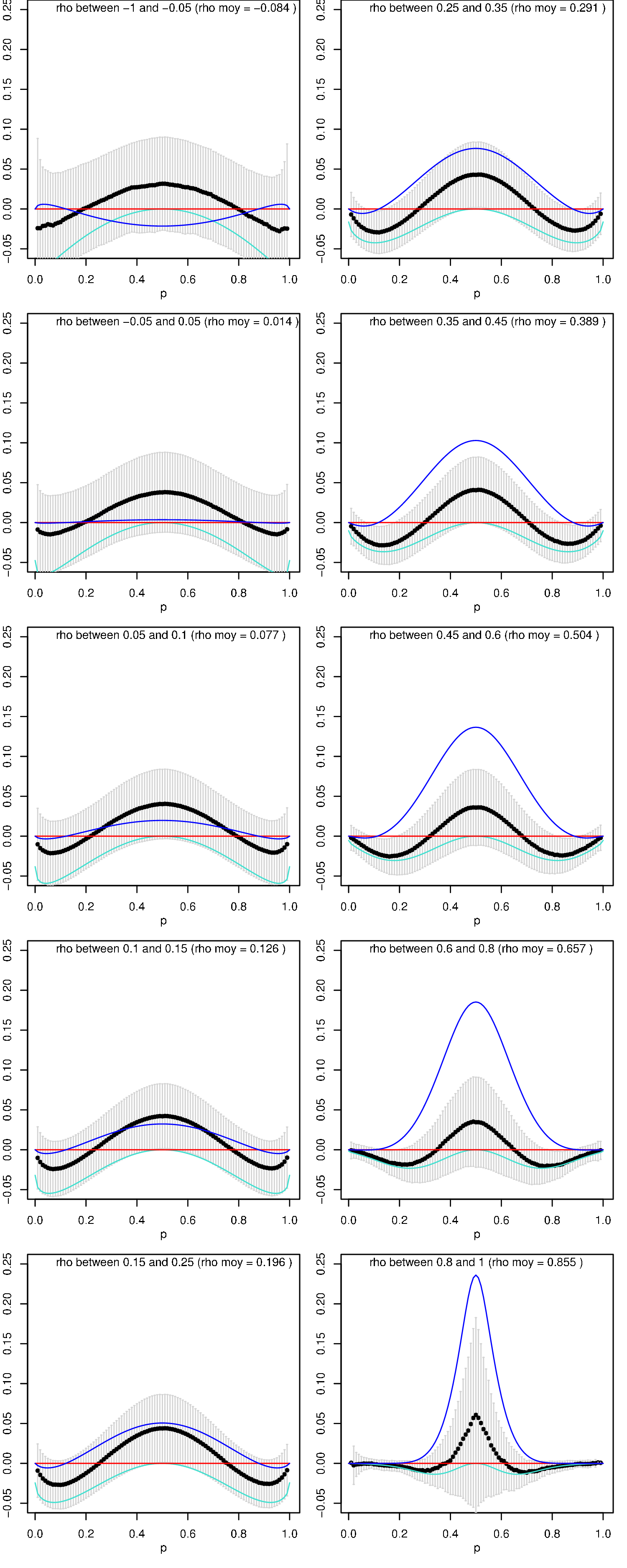}
  \caption{Diagonal (left) and anti-diagonal (right) of the copula vs $p$.
        Real data (2000-2004) averaged in 10 bins of $\rho$. The vertical lines corresponds to 1-sigma dispersion of the results (i.e. not 
        the statistical error bar).
	The turquoise line corresponds to an elliptical model (Student with $\nu=5$):
	it predicts too low a value at the centre point and too large tail dependences (left and right limits).
	The blue line depicts the behaviour of Frank's copula.
	It is worth noticing that the tail dependence coefficient is always measured empirically at some $p<1$, for example $p=0.95$.
	In this case, the disagreement between data and an elliptical model may be accidentally small.
	It is important to distinguish the interpretation of the limit coefficient from that of a penultimate value [6].}
  \label{fig:cop_dev}
\end{figure}

\subsection{Pseudo-elliptical log-normal model}\label{ssec:pseudo_ell_log}

As we just showed, the fact that the central value of the copula $C^*$ does not obey the relation 
$-\cos \left(2\pi C^*\right) = \rho$ rules out all elliptical models. 
One possible way out is to consider a model where random volatilities are stock dependent, as explained in Sect.~\ref{ssec:pseudo-ell}. 
Choosing for simplicity a log-normal model of correlated volatilities and inserting \eqref{eq:fd_lognorm} into \eqref{eq:coeffs_pseudo_ell}, the new prediction is:
\be
\rho = - \e^{s^2(c - 1)} \cos\left(2\pi C^*\right),
\ee
where $c$ is the correlation of the log-volatilities. 
This suggests to plot \mbox{$z=\ln \left(\rho/|\cos\left(2\pi C^*\right)|\right)$} as a function of $\rho$, as shown in Fig.~\ref{fig:C05b}. 
For a purely elliptical model, $z$ should be identically zero, corresponding to perfectly correlated volatilities ($c=1$).
What we observe in Fig.~\ref{fig:C05b}, on the other hand, 
is that $c$ is indeed close to unity for large enough $\rho$'s, 
but systematically decreases as $\rho$ decreases; 
in other words, the volatilities of weakly correlated stocks are themselves weakly correlated. 
This is, again, in line with the conclusion we reached above. 

However, this pseudo-elliptical model still predicts that $C^* = \frac14$ for stocks with zero linear correlations, 
in disagreement with the data shown in Fig.~\ref{fig:C05a}. 
This translates, in Fig.~\ref{fig:C05b}, into a negative divergence of $z$ when $\rho \to 0$. 
This finding means that we should look for other types of constructions. 
How can one have at the same time $\rho = 0$ and $C^* > \frac14$ ? 
A toy-model, that serves as the basis for a much richer model that we will report elsewhere \cite{remi}, is the following. 
Consider two independent, symmetrically distributed random factors $\psi_1$ and $\psi_2$ with equal volatilities, 
and construct the returns of assets $1$ and $2$ as:
\be
X_{1,2} = \psi_1 \pm \psi_2.
\ee
Clearly, the linear correlation is zero. 
Using a cumulant expansion, one easily finds that to the kurtosis order, the central value of the copula is equal to:
\be
C^* = \frac14 + \frac{\kappa_2 - \kappa_1}{24 \pi} + \dots
\ee
Therefore, if the kurtosis $\kappa_1$ of the factor to which both stocks are positively exposed is smaller than the kurtosis $\kappa_2$ of the spread, 
one does indeed find $C^*(\rho=0) > \frac14$. 

We will investigate in a subsequent paper \cite{remi} additive models such as the above one, with random volatilities affecting the factors $\psi_1$ and $\psi_2$, that
generalise elliptical models in a way to capture the presence of several volatility modes.

\section{Conclusion}\label{sec:conclusion}

The object of this paper was to discuss the adequacy of Student copulas, or more generally elliptical copulas, 
to describe the multivariate distribution of stock returns. 
We have suggested several ideas of methodological interest to efficiently visualise and  compare different copulas. 
We recommend in particular the rescaled difference with the Gaussian copula along the diagonal and 
the central value of the copula as strongly discriminating observables. 
We have studied the dependence of these quantities, as well as other non-linear correlation coefficients, 
with the linear correlation coefficient $\rho$, and made explicit the predictions of elliptical models that can be empirically tested. 
We have shown, using a very large data set, with the daily returns of 1500 US stocks over 15 years, 
that elliptical models fail to capture the detailed structure of stock dependences. 

In a nutshell, the main message elicited by our analysis is that Student copulas provide a good approximation 
to describe the joint distribution of strongly correlated pairs of stocks, but badly miss their target for weakly correlated stocks. 
We believe that the same results hold for a wider class of assets: 
it is plausible that highly correlated assets do indeed share the same risk factor. 
Intuitively, the failure of elliptical models to describe weakly correlated assets can be traced to the 
inadequacy of the assumption of a single ``market''  volatility mode.
We expect that exactly as for returns, several factors are needed to capture sectorial volatility modes and idiosyncratic modes as well. 
The precise way to encode this idea into a workable model that would naturally generalise elliptical models and faithfully capture the non-linear, 
asymmetric dependence in stock markets is at this stage an open problem, on which we will report in an upcoming paper. 
We strongly believe that such a quest cannot be based on a formal construction of mathematically convenient models 
(such as Archimedean copulas that, in our opinion, cannot be relevant to describe asset returns). 
The way forward is to rely on intuition and plausibility and come up with models that make financial sense. 
This is, of course, not restricted to copula modeling, but applies to all quarters of quantitative finance.

\appendix
\section{Proof of Eqn.~\eqref{eq:tau_puissance}, page~\pageref{eq:tau_puissance}}\label{apx:tau_puissance}

\begin{lemma}
Let $(X_1,X_2)$ follow a bivariate student distribution with $\nu$ degrees of freedom and correlation $\rho$.
Denote by $T_{\nu}(x)$ the univariate Student cdf with $\nu$ degrees of freedom,
and $x_p\equiv{}T_{\nu}^{-1}(p)$ its $p$-quantile, and define
\begin{subequations}\label{eq:K_def1}
\begin{align}
	K   &=\sqrt{\frac{\nu+1}{1-\rho^2}}\frac{X_1-\rho{}x_p}{\sqrt{\nu+X_2^2}}\\
	k(p)&=\frac{\sqrt{\nu+1}\sqrt{1-\rho}}{\sqrt{1+\rho}}\left(\nu\cdot{}x_p^{-2}+1\right)^{-\frac{1}{2}}
	     =\frac{k(1)}{\sqrt{1+\frac{\nu}{x_p^{2}}}}
\end{align}
\end{subequations}

Then, 
\[
	\mathds{P}[K\leq k(p)\mid X_2=x_p]=T_{\nu+1}\big(k(p)\big)
\]
\end{lemma}
\begin{proof}
The proof procedes straightforwardly by showing that $t_{\nu+1}(k)=\mathcal{P}_{X_1|X_2}(x_1|x_2)\frac{\partial{}x_1}{\partial{}k}$ when $x_2=x_p$.
The particular result for the limit case $p=1$ is stated in \cite{embrechts02correlation}.
\end{proof}

\begin{theorem}\footnote{
This result was simultaneously found by Manner and Segers \cite{manner2009tails} (in a somewhat more general context). 
We still sketch our proof because it follows a different route (uses the Copula), and the final expression looks quite different, although of course numerically identical.
}
Let $(X_1,X_2)$ be defined like in the Lemma above.
The pre-asymptotic behaviour of its tail dependence when $p\to 1$ is approximated by the following expansion in (rational) powers of $(1-p)$:
\begin{equation}\tag{\ref{eq:tau_puissance}}
	\tUU(p)=\tUU^*(\nu,\rho) + \beta(\nu,\rho)\cdot (1-p)^{\frac{2}{\nu}} + \mathcal{O}\big((1-p)^{\frac{4}{\nu}}\big)
\end{equation}
with
\begin{align}
	\tUU^*(\nu,\rho)&=2-2\:T_{\nu+1}(k(1))\\
	 \beta(\nu,\rho)&=\frac{\nu^{\frac{2}{\nu}+1}}{\frac{2}{\nu}+1}k(1)\,t_{\nu+1}\big(k(1)\big)\,L_{\nu}^{-\frac{2}{\nu}}\\\nonumber
	                &=\left(\frac{\Gamma(\frac{\nu}{2})\sqrt{\pi}}{\Gamma(\frac{\nu+1}{2})}\right)^{\frac{2}{\nu}}
		          \frac{\nu^{\frac{2}{\nu}}}{\frac{2}{\nu}+1}
		          \frac{\sqrt{\nu+1}\sqrt{1-\rho}}{\sqrt{1+\rho}}\cdot{}t_{\nu+1}\left(\frac{\sqrt{\nu+1}\sqrt{1-\rho}}{\sqrt{1+\rho}}\right)
\end{align}
\end{theorem}
\begin{proof}

Recall from Eq.~\eqref{eq:tau_C} that
\begin{equation}\label{eq:app1}
	\tUU(p)=2-\frac{C(p,p)-C(1,1)}{p-1}=2-\frac{1}{1-p}\int_p^1\frac{dC(p,p)}{dp}dp
\end{equation}
One easily shows that
\begin{align*}
	dC(p,p)&=2\:\mathds{P}[X_1\leq{}x_p\mid{}X_2=x_p]dp=2\:T_{\nu+1}\big(k(p)\big)dp
\end{align*}
where the second equality holds in virtue of the aforementioned lemma.\newline
Now, for $p$ close to $1$, $k(p)=k(1)\left(1-\frac{1}{2}\frac{\nu}{x_p^2}\right)+\mathcal{O}(x_p^{-4})$, and
\begin{equation*}
	\frac{dC(p,p)}{dp}=2\,T_{\nu+1}\big(k(p)\big)\approx{}2\,T_{\nu+1}\big(k(1)\big)-2\big(k(1)-k(p)\big)\cdot{}t_{\nu+1}\big(k(1)\big)
\end{equation*}
But since the Student distribution behaves as a power-law precisely in the region $p\lesssim 1$, we write $t_{\nu}(x)\approx L_{\nu}/x^{\nu+1}$ and immediately get
\begin{equation*}
	1-p=\frac{L_{\nu}}{\nu}x_p^{-\nu}\qquad\text{with}\quad L_{\nu}=\frac{\Gamma(\frac{\nu+1}{2})}{\Gamma(\frac{\nu}{2})\sqrt{\pi}}\nu^{\frac{\nu}{2}}
\end{equation*}
The result follows by collecting all the terms and performing the integration in \eqref{eq:app1}.
\end{proof}

Notice that for large $\nu$, the exponent is almost zero and the correction term is of order $\mathcal{O}(1)$, and the expansion ceases to hold.
This is particularly true for the Gaussian distribution ($\nu=\infty$) for which the behaviour is radically different at the limit $p\to 1$
(where there's strictly no tail correlation) and at $p<1$ where a dependence subsists, see Fig.~\ref{fig:tau_stud1}.

\section*{Acknowledgments}

We want to thank F.~Abergel, R.~Allez, S.~Ciliberti, L.~Laloux \& M.~Potters for help and useful comments.

\nocite{coles1999dependence}
\nocite{klein2005capital}
\nocite{dempster2002risk}

\bibliographystyle{ijtaf}
\bibliography{../biblio_all}

\end{document}